\documentclass[11pt]{article}
 
\usepackage{amsmath,amssymb,amsthm,graphicx}

\newcommand{\postbqp}{{\sf PostBQP}}
\newcommand{\postbpp}{{\sf PostBPP}}
\newcommand{\ph}{{\sf PH}}
\newcommand{\p}{{\sf P}}

\newtheorem{definition}{Definition}
\newtheorem{lemma}{Lemma}
\newtheorem{theorem}{Theorem}

\begin{document}

\title{Commuting Quantum Circuits with Few Outputs are\\
Unlikely to be Classically Simulatable}

\author{Yasuhiro Takahashi\\
{\small NTT Communication Science Laboratories,}\\
{\small NTT Corporation}\\
{\footnotesize\tt takahashi.yasuhiro@lab.ntt.co.jp}
\and
Seiichiro Tani\\
{\small NTT Communication Science Laboratories,}\\
{\small NTT Corporation}\\
{\footnotesize\tt tani.seiichiro@lab.ntt.co.jp}
\and
Takeshi Yamazaki\\
{\small Mathematical Institute, Tohoku University}\\
{\footnotesize\tt yamazaki@math.tohoku.ac.jp}
\and
Kazuyuki Tanaka\\
{\small Mathematical Institute, Tohoku University}\\
{\footnotesize\tt tanaka@math.tohoku.ac.jp}}




\date{}

\maketitle

\begin{abstract}
 We study the classical simulatability of commuting quantum circuits
 with $n$ input qubits and $O(\log n)$ output qubits, where a quantum
 circuit is classically simulatable if its output probability
 distribution can be sampled up to an exponentially small additive error
 in classical polynomial time. First, we show that there exists a
 commuting quantum circuit that is not classically simulatable unless
 the polynomial hierarchy collapses to the third level. This is the
 first formal evidence that a commuting quantum circuit is not
 classically simulatable even when the number of output qubits is
 exponentially small. Then, we consider a generalized version of the
 circuit and clarify the condition under which it is classically
 simulatable. Lastly, we apply the argument for the above evidence to
 Clifford circuits in a similar setting and provide evidence that such a
 circuit augmented by a depth-1 non-Clifford layer is not classically
 simulatable. These results reveal subtle differences between quantum
 and classical computation.
\end{abstract}

\section{Introduction and Summary of Results}

One of the most important challenges in quantum information processing
is to understand the difference between quantum and classical
computation. An approach to meeting this challenge is to study the
classical simulatability of quantum computation. Previous studies have
shown that restricted models of quantum computation, such as commuting
quantum circuits, are useful for this
purpose~\cite{Terhal,Fenner,Shepherd2,Shepherd,Aaronson1,Bremner,Ni,Jozsa,Takahashi-sim,Morimae}. Because of the simplicity of such restricted models,
they are also useful for identifying the source of the computational
power of quantum computers. It is therefore of great interest to study
their classical simulatability.

In this paper, we study the classical simulatability of commuting
quantum circuits with $n$ input qubits and $O({\rm poly}(n))$ ancillary
qubits initialized to $|0\rangle$, where a commuting quantum circuit is
a quantum circuit consisting of pairwise commuting gates, each of which
acts on a constant number of qubits. When all commuting gates in a
commuting quantum circuit act on at most $c$ qubits for some constant
$c\geq 2$, the circuit is said to be $c$-local. For considering the
classical simulatability, we adopt strong and weak simulations. The
strong simulation of a quantum circuit is to compute its output
probability up to an exponentially small additive error in classical
polynomial time and the weak one is to sample its output probability
distribution similarly. Any strongly simulatable quantum circuit is
weakly simulatable. Our main focus is on the hardness of classically
simulating quantum circuits and thus we mainly deal with the weak
simulatability, which yields a stronger result than that the strong
simulatability yields. Previous hardness results on the weak
simulatability are usually obtained with respect to multiplicative
error~\cite{Terhal,Bremner,Jozsa}, but such an error seems to be too
strong an assumption as discussed in~\cite{Aaronson1}. Our results are
obtained with respect to additive error.

In 2011, Bremner et al.\ showed that there exists a 2-local IQP circuit
with $O({\rm poly}(n))$ output qubits such that it is not weakly
simulatable (under a plausible assumption)~\cite{Bremner}, where an IQP
circuit is a quantum circuit consisting of pairwise commuting gates that
are diagonal in the $X$-basis $\{(|0\rangle\pm
|1\rangle)/\sqrt{2}\}$. Roughly speaking, this result means that when
the number of output qubits is large, even a simple commuting
quantum circuit is powerful. On the other hand, in 2013, Ni et al.\
showed that any 2-local commuting quantum circuit with $O(\log n)$
output qubits is strongly simulatable and that there exists a 3-local
commuting quantum circuit with only one output qubit such that it is not
strongly simulatable (under a plausible assumption)~\cite{Ni}. Thus,
when the number of output qubits is $O(\log n)$, the classical
simulatability of commuting quantum circuits depends on the number of
qubits affected by each commuting gate. A natural question is whether
there exists a commuting quantum circuit with $O(\log n)$ output qubits
such that it is not weakly simulatable.

There are two previous results related to this question. The first one
is that any (constant-local) IQP circuit with $O(\log n)$ output qubits
is weakly simulatable~\cite{Bremner}. Thus, if we want to answer the
above question affirmatively, we need to consider commuting quantum
circuits other than IQP circuits. The second one is that, if any
commuting quantum circuit with only one output qubit is weakly
simulatable, there exists a polynomial-time classical algorithm for the
problem of estimating the matrix element $|\langle 0| U |0\rangle|$
(up to a polynomially small additive error) for any unitary matrix $U$
that is implemented by a constant-depth quantum circuit~\cite{Ni}. This
suggests an affirmative answer to the above question since the matrix
element estimation problem seems to be hard for a classical
computer. However, the hardness has not been formally understood yet.

We provide the first formal evidence for answering the above question
affirmatively:
\begin{theorem}\label{commuting}
There exists a 5-local commuting quantum circuit with $O(\log n)$ output
 qubits such that it is not weakly simulatable unless the polynomial
 hierarchy $\ph$ collapses to the third level.
\end{theorem}
\noindent
It is widely believed that $\ph$ does not collapse to any
level~\cite{Papadimitriou}. Thus, the circuit in Theorem~\ref{commuting} is the
desired evidence. To construct the circuit, we first show the existence
of a depth-3 quantum circuit $A_n$ that is not weakly simulatable with
respect to additive error (under a plausible assumption), where it has
$n$ input qubits, $O({\rm poly}(n))$ ancillary qubits, and $O({\rm
poly}(n))$ output qubits. This is shown by our new analysis of the weak
simulatability (with respect to additive error) of a depth-3 quantum
circuit that is not weakly simulatable with respect to multiplicative
error (under a plausible assumption)~\cite{Bremner,Fenner}. Our idea for
constructing the circuit in Theorem~\ref{commuting} is to combine $A_n$
with the OR reduction circuit~\cite{Hoyer}, which reduces the
computation of the OR function on $k$ bits to that on $O(\log k)$
bits. The resulting circuit has $O(\log n)$ output qubits and is not
weakly simulatable (under a plausible assumption). It is of course not a
commuting quantum circuit, but an important observation is that the OR
reduction circuit can be transformed into a 2-local commuting quantum
circuit. We consider a quantum circuit consisting gates of the form
$A_n^\dag g A_n$ for any commuting gate $g$ in the commuting OR
reduction circuit and analyze it rigorously, which implies
Theorem~\ref{commuting}.

Then, in order to generalize the above-mentioned result that any IQP
circuit with $O(\log n)$ output qubits is weakly
simulatable~\cite{Bremner}, we consider the weak simulatability of a
generalized version of the circuit in Theorem~\ref{commuting}. We assume
that we are given two quantum circuits: $F_n$ is a quantum circuit with
$n$ input qubits, $O({\rm poly}(n))$ ancillary qubits, and $O({\rm
poly}(n))$ output qubits and $D$ is a quantum circuit on $O({\rm
poly}(n))$ qubits consisting of pairwise commuting gates that are
diagonal in the $Z$-basis $\{|0\rangle,|1\rangle\}$. The generalized
version is the circuit $(F_n^\dag\otimes H^{\otimes l})D(F_n\otimes
H^{\otimes l})$, where $l=O(\log n)$. The input qubits and output qubits
of the circuit are the input qubits of $F_n$ and the ancillary qubits on
which $H^{\otimes l}$ is applied, respectively. In particular, when
$F_n=A_n$ and $D$ is a quantum circuit consisting of controlled
phase-shift gates, the whole circuit becomes the circuit in
Theorem~\ref{commuting}. We show that the weak simulatability of $F_n$
implies that of the whole circuit:
\begin{theorem}\label{sois}
 If $F_n$ is weakly simulatable, then $(F_n^\dag\otimes H^{\otimes
 l})D(F_n\otimes H^{\otimes l})$ with $l=O(\log n)$ output qubits is
 also weakly simulatable.
\end{theorem}
\noindent
The above-mentioned result in~\cite{Bremner} corresponds to the case
when $F_n$ is a tensor product of $H$. Theorem~\ref{sois} implies an
interesting suggestion on how to improve Theorem~\ref{commuting}. As
described above, the 5-local commuting quantum circuit in
Theorem~\ref{commuting} is constructed by choosing a depth-3 quantum
circuit as $F_n$. A possible way to improve Theorem~\ref{commuting}, or
more concretely, a possible way to construct a 3- or 4-local commuting
quantum circuit that is not weakly simulatable would be to somehow
choose a depth-2 quantum circuit as $F_n$. Theorem~\ref{sois} implies
that such a construction is impossible. This is because, since any
depth-2 quantum circuit is weakly simulatable~\cite{Terhal,Markov},
choosing a depth-2 quantum circuit as $F_n$ yields only a weakly
simulatable quantum circuit.

We show Theorem~\ref{sois} by simply generalizing the proof of the
above-mentioned result in~\cite{Bremner}. More precisely, we fix the
states of the qubits other than the $O(\log n)$ output qubits on the
basis of the assumption in Theorem~\ref{sois} and then follow the change
of the states of the output qubits. This yields a polynomial-time
classical algorithm for weakly simulating $(F_n^\dag\otimes H^{\otimes
l})D(F_n\otimes H^{\otimes l})$.

Lastly, we apply the argument for proving Theorem~\ref{commuting}
to Clifford circuits with $n$ input qubits, $O({\rm poly}(n))$ ancillary
qubits in a product state, and $O(\log n)$ output qubits. A simple
extension of the proof in~\cite{Clark,Jozsa} implies that any Clifford
circuit in the setting is strongly simulatable. We provide evidence that
a slightly extended circuit is not weakly simulatable:
\begin{theorem}\label{Clifford-1}
 There exists a Clifford circuit augmented by a depth-1 non-Clifford
 layer with $O({\rm poly}(n))$ ancillary qubits in a particular product
 state and with $O(\log n)$ output qubits such that it is not weakly
 simulatable unless $\ph$ collapses to the third level.
\end{theorem}
\noindent
Similar to Theorems~\ref{commuting} and~\ref{sois},
Theorem~\ref{Clifford-1} contributes to understanding a subtle
difference between quantum and classical computation. As in the proof of
Theorem~\ref{commuting}, using the result in~\cite{Jozsa}, we show the
existence of a Clifford circuit that is not weakly simulatable with
respect to additive error (under a plausible assumption), where it has
$n$ input qubits, $O({\rm poly}(n))$ ancillary qubits in a particular
product state, and $O({\rm poly}(n))$ output qubits. Then, we combine
the Clifford circuit with a constant-depth OR reduction circuit with
unbounded fan-out gates~\cite{Hoyer}. The resulting circuit has $O(\log
n)$ output qubits and is not weakly simulatable (under a plausible
assumption). By decomposing the unbounded fan-out gates into CNOT gates,
we transform the combination of the Clifford circuit and OR reduction
circuit into a Clifford circuit augmented by a depth-1 non-Clifford
layer, which implies Theorem~\ref{Clifford-1}. A similar argument with a
constant-depth quantum circuit for the OR function with unbounded
fan-out gates~\cite{Takahashi} implies that the number of output qubits
can further be decreased to one at the cost of adding one more depth-1
non-Clifford layer.

\section{Preliminaries}

\subsection{Quantum Circuits}

We use the standard notation for quantum states and the standard
diagrams for quantum circuits~\cite{Nielsen}. The elementary gates in
this paper are a Hadamard gate $H$, a phase-shift gate $R(\theta)$ with
angle $\theta =\pm 2\pi/2^k$ for any $k \in {\mathbb N}$, and a
controlled-$Z$ gate $\Lambda Z$, where
$$
H=\frac{1}{\sqrt{2}}
\begin{pmatrix}
1 & 1 \\
1 & -1
\end{pmatrix},\
R(\theta)=
\begin{pmatrix}
1 & 0 \\
0 & e^{i\theta}
\end{pmatrix},\
\Lambda Z=
\begin{pmatrix}
1 & 0 & 0 & 0 \\
0 & 1 & 0 & 0 \\
0 & 0 & 1 & 0 \\
0 & 0 & 0 & -1
\end{pmatrix}.
$$
We denote $R(\pi)$, $R(\pi/2)$, and $HR(\pi)H$ as $Z$, $P$, and $X$,
respectively, where $Z$ and $X$ (with $Y=iXZ$ and identity $I$) are
called Pauli gates. We also denote $H\Lambda Z H$ as $\Lambda X$, 
which is a CNOT gate, where $H$ acts on the target qubit. A quantum
circuit consists of the elementary gates. In particular, when a quantum
circuit consists only of $H$, $P$, and $\Lambda Z$, it is called a
Clifford circuit. A commuting quantum circuit is a quantum circuit
consisting of pairwise commuting gates, where we do not require that
each commuting gate be one of the elementary gates. In other words, when
we think of a quantum circuit as a commuting quantum circuit, we are
allowed to regard a group of elementary gates in the circuit as a single
gate and we require that such gates, which are not necessarily
elementary gates, be pairwise commuting.

The complexity measures of a quantum circuit are its size and depth. The
size is the number of elementary gates in the circuit. To define the
depth, we consider the circuit as a set of layers $1,\ldots,d$
consisting of one-qubit and two-qubit gates, where gates in the same
layer act on pairwise disjoint sets of qubits and any gate in layer $j$
is applied before any gate in layer $j+1$. The depth of the circuit is
the smallest possible value of $d$~\cite{Fenner}. It seems to be natural
to require that each gate in a layer be one of the elementary gates,
but we do not require this for simplicity and we consider one-qubit and
two-qubit gates determined from the context. In other words, when we
count the depth, we are allowed to consider one-qubit and two-qubit
gates generated by elementary gates in the circuit. Regardless of
whether we adopt the requirement or not, the depth of the circuit we are
interested in is a constant. A quantum circuit can use ancillary qubits
initialized to $|0\rangle$. We do not require that the states of the
ancillary qubits be reset to $|0\rangle$ at the end of the computation.

We deal with a uniform family of polynomial-size quantum circuits
$\{C_n\}_{n\geq 1}$, where each $C_n$ is a quantum circuit with $n$
input qubits and $O({\rm poly}(n))$ ancillary qubits, and can use
phase-shift gates with angles $\theta=\pm2\pi/2^k$ for any $k=O({\rm
poly}(n))$. Some of the input and ancillary qubits are called output
qubits. At the end of the computation, $Z$-measurements, i.e.,
measurements in the $Z$-basis, are performed on the output qubits. The
uniformity means that there exists a polynomial-time deterministic 
classical algorithm for computing the function $1^n \mapsto
\overline{C_n}$, where $\overline{C_n}$ is the classical description of
$C_n$. A symbol denoting a quantum circuit, such as $C_n$, also denotes
its matrix representation in some fixed basis. Any quantum circuit in
this paper is understood to be an element of a uniform family of
polynomial-size quantum circuits and thus, for simplicity, we deal with
a quantum circuit $C_n$ in place of a family $\{C_n\}_{n\geq 1}$. We
require that each commuting gate in a commuting quantum circuit act on a
constant number of qubits. When all commuting gates act on at most $c$
qubits for some constant $c\geq 2$, the circuit is said to be
$c$-local~\cite{Ni}.

\subsection{Classical Simulatability and Complexity Classes}

We deal with a uniform family of polynomial-size classical circuits to
 model a polynomial-time deterministic classical algorithm. Similarly,
 to model its probabilistic version, we deal with a uniform family of
 polynomial-size randomized classical circuits, each of which has a
 register initialized with random bits for each run of the
 computation~\cite{Bremner}. As in the case of quantum circuits, for
 simplicity, we consider a classical circuit in place of a family of
 classical circuits.

 Let $C_n$ be a polynomial-size quantum circuit with $n$ input qubits,
 $O({\rm poly}(n))$ ancillary qubits, and $m$ output qubits. For any $x
 \in \{0,1\}^n$, there exists an output probability distribution
 $\{(y,{\rm Pr}[C_n(x)=y])\}_{y \in \{0,1\}^m}$, where ${\rm
 Pr}[C_n(x)=y]$ is the probability of obtaining $y\in \{0,1\}^m$ by
 $Z$-measurements on the output qubits of $C_n$ with the input state
 $|x\rangle$. The classical simulatability of $C_n$ is defined as
 follows~\cite{Terhal,van1,Bremner,van2,Ni,Jozsa,Takahashi-sim}:
\begin{definition}\label{sim}
\begin{itemize}
 \item $C_n$ is strongly simulatable if the output
       probability ${\rm Pr}[C_n(x)=y]$ and its marginal output
       probability can be computed up to an exponentially small additive
       error in classical $O({\rm poly}(n))$ time. More precisely, for
       any polynomial $p$, there exists a polynomial-size classical
       circuit $D_n$ such that, for any $x\in\{0,1\}^n$ and
       $y\in\{0,1\}^m$, 
       $$|D_n(x,y)-{\rm Pr}[C_n(x)=y]| \leq \frac{1}{2^{p(n)}},$$
       and, when we choose arbitrary $m'$ output qubits from the $m$
       output qubits of $C_n$ for any $m'< m$, the output probability
       ${\rm Pr}[C_n(x)=y']$ can be computed similarly for any
       $x\in\{0,1\}^n$ and $y'\in \{0,1\}^{m'}$.

 \item $C_n$ is weakly simulatable if the output
       probability distribution $\{(y,{\rm Pr}[C_n(x)=y])\}_{y \in
       \{0,1\}^m}$ can be sampled up to an exponentially small additive
       error in classical $O({\rm poly}(n))$ time. More precisely, for
       any polynomial $p$, there exists a polynomial-size randomized
       classical circuit $R_n$ such that, for any $x\in\{0,1\}^n$ and
       $y\in\{0,1\}^m$,
       $$|{\rm Pr}[R_n(x)=y]-{\rm Pr}[C_n(x)=y]| \leq
       \frac{1}{2^{p(n)}}.$$
\end{itemize}
\end{definition}
\noindent
Any strongly simulatable quantum circuit is weakly
simulatable~\cite{Terhal,Bremner}.

The following two complexity classes are important for our
discussion~\cite{Aaronson0,Bremner,Han}:
\begin{definition}
Let $L \subseteq \{0,1\}^*$.
\begin{itemize}
 \item $L \in \postbqp$ if there exists a polynomial-size quantum circuit
       $C_n$ with $n$ input qubits, $O({\rm poly}(n))$ ancillary qubits,
       one output qubit, and one particular qubit (other than the output
       qubit) called the postselection qubit such that, for any $x \in
       \{0,1\}^n$,
\begin{itemize}
 \item  ${\rm Pr}[{\rm post}_n(x)=0] > 0$,

 \item if $x \in L$, ${\rm Pr}[C_n(x)=1|{\rm post}_n(x)=0] \geq 2/3$,

 \item if $x \notin L$, ${\rm Pr}[C_n(x)=1|{\rm post}_n(x)=0] \leq
       1/3$,
 \end{itemize}
where the event ``${\rm post}_n(x)=0$'' means that the classical outcome
       of the $Z$-measurement on the postselection qubit is 0.
 \item $L \in \postbpp$ if there exists a polynomial-size randomized
       classical circuit $R_n$ with $n$ input bits that, for any $x
       \in \{0,1\}^n$, outputs $R_n(x),{\rm post}_n(x) \in \{0,1\}$ such
       that
\begin{itemize}
 \item ${\rm Pr}[{\rm post}_n(x)=0] > 0$,

 \item if $x \in L$, ${\rm Pr}[R_n(x)=1|{\rm post}_n(x)=0] \geq 2/3$,

 \item if $x \notin L$, ${\rm Pr}[R_n(x)=1|{\rm post}_n(x)=0] \leq
       1/3$.
 \end{itemize}
\end{itemize}
\end{definition}
\noindent
We use the notation ${\rm post}_n(x)=0$ both in the quantum and
classical settings, but the meaning will be clear from the
context. Another important class is the polynomial hierarchy $\ph =
\bigcup_{j\geq 1} \Delta_j^p$. Here, $\Delta_1^p = \p$ and
$\Delta_{j+1}^p = \p^{{\sf N}\Delta_j^p}$ for any $j\geq 1$, where $\p$
is the class of languages decided by polynomial-time classical
algorithms and ${{\sf N}\Delta_j^p}$ is the non-deterministic class
associated to $\Delta_j^p$~\cite{Papadimitriou,Bremner}. It is widely
believed that $\ph \neq \Delta_j^p$ for any $j\geq
1$~\cite{Papadimitriou}. As shown in~\cite{Bremner}, if $\postbqp
\subseteq \postbpp$, then $\ph = \Delta_3^p$, i.e., $\ph$ collapses to
the third level. It can be shown that, in our setting of elementary
gates and quantum circuits, this relationship also holds when the
condition ${\rm Pr}[{\rm post}_n(x)=0] > 0$ in the definition of
$\postbqp$ is replaced with the condition that, for some polynomial $q$
(depending only on $C_n$), ${\rm Pr}[{\rm post}_n(x)=0] \geq
1/2^{q(n)}$. In the following, we adopt the latter condition.

\section{Commuting Quantum Circuits}

\subsection{Hardness of the Weak Simulation}

It is known that there exists a depth-3 quantum circuit with $n$ input
qubits, $O({\rm poly}(n))$ ancillary qubits, and $O({\rm poly}(n))$
output qubits such that it is not weakly simulatable with respect to
multiplicative error unless $\ph$ collapses to the third
level~\cite{Bremner}. We first analyze its weak simulatability with
respect to additive error and show the following lemma:
\begin{lemma}\label{depth-3}
There exists a depth-3 polynomial-size quantum circuit with $O({\rm
 poly}(n))$ output qubits such that it is not weakly simulatable (with
 respect to additive error) unless $\ph$ collapses to the third level.
\end{lemma}
\begin{proof}
We assume that $\ph$ does not collapse to the third level. Then, as
 described above, $\postbqp \nsubseteq \postbpp$. Let $L \in
 \postbqp\setminus \postbpp$. Then, there exists a polynomial-size 
 quantum circuit $C_n$ with $n$ input qubits, $a=O({\rm poly}(n))$
 ancillary qubits, one output qubit, and one postselection qubit (and
 some polynomial $q$) such that, for any $x \in \{0,1\}^n$,
\begin{itemize}
 \item ${\rm Pr}[{\rm post}_n(x)=0] \geq 1/2^{q(n)}$,

 \item if $x \in L$, ${\rm Pr}[C_n(x)=1|{\rm post}_n(x)=0] \geq 2/3$,

 \item if $x \notin L$, ${\rm Pr}[C_n(x)=1|{\rm post}_n(x)=0] \leq
       1/3$.
\end{itemize}
As shown in~\cite{Fenner}, there exists a depth-3 polynomial-size
 quantum circuit $A_n$ with $n$ input qubits, $a+b$ ancillary qubits,
 and one output qubit such that, for any $x \in \{0,1\}^n$,
\begin{itemize}
 \item if $x \in L$, Pr$[A_n(x)=1|{\rm qpost}_n(x)=0^{b+1}] \geq 2/3$,

 \item if $x \notin L$, Pr$[A_n(x)=1|{\rm qpost}_n(x)=0^{b+1}] \leq
       1/3$,
\end{itemize}
where $b=O({\rm poly}(n))$, the event ``${\rm qpost}_n(x)=0^{b+1}$''
 means that all classical outcomes of $Z$-measurements on the qubit
 corresponding to the postselection qubit of $C_n$ and particular $b$
 qubits (other than the output qubit) are 0. We call these $b+1$ qubits
 the postselection qubits of $A_n$. Since the probability of obtaining
 $0^b$ by $Z$-measurements on the $b$ qubits is $1/2^b$~\cite{Fenner},
 it holds that
$${\rm Pr}[{\rm qpost}_n(x)=0^{b+1}]=\frac{1}{2^{b}}\cdot {\rm Pr}[{\rm
 post}_n(x)=0] \geq \frac{1}{2^{b+q}}.$$
We regard $A_n$, which has only one output qubit, as a new circuit with
 $b+2$ output qubits, where one of the output qubits is the original
 output qubit $q_{\rm out}$ of $A_n$ and the others are the $b+1$
 postselection qubits of $A_n$. We also denote this circuit as
 $A_n$. Thus, $A_n$ is a depth-3 polynomial-size quantum circuit with
 $O({\rm poly}(n))$ output qubits. For any $x \in \{0,1\}^n$,
\begin{itemize}
 \item ${\rm Pr}[A_n(x)=0^{b+1}1]={\rm Pr}[A_n(x)=1\&{\rm
       qpost}_n(x)=0^{b+1}]$,
 \item ${\rm Pr}[A_n(x)=0^{b+1}0]={\rm Pr}[A_n(x)=0\&{\rm
       qpost}_n(x)=0^{b+1}]$,
\end{itemize}
where, for simplicity, we assume that the last output qubit of $A_n$ is
 $q_{\rm out}$. Thus, for any $x \in \{0,1\}^n$,
\begin{itemize}
 \item if $x \in L$, ${\rm Pr}[A_n(x)=0^{b+1}1] \geq 2\cdot{\rm
       Pr[qpost}_n(x)=0^{b+1}]/3$,
 \item if $x \notin L$, ${\rm Pr}[A_n(x)=0^{b+1}1] \leq {\rm
       Pr[qpost}_n(x)=0^{b+1}]/3$.
\end{itemize}
We can show that, if $A_n$ is weakly simulatable, then $L \in
 \postbpp$. This contradicts the assumption that $L\notin \postbpp$ and
 completes the proof. The details can be found in Appendix~A.1.
\end{proof}
The proof method of Lemma~\ref{depth-3} can be considered as an
elaborated version of the one in~\cite{Takahashi-sim}. As pointed out by
Nishimura and Morimae~\cite{Nishimura}, we note that their proof method
in~\cite{Morimae} based on the complexity class $\sf
SBQP$~\cite{Kuperberg} can also be used to show the lemma.

The OR reduction circuit reduces the computation of the OR
function on $b$ bits to that on $O(\log b)$ bits~\cite{Hoyer}: for any
$b$-qubit input state $|x_1\rangle\cdots |x_b\rangle$ with
$x_j\in\{0,1\}$, the circuit outputs $|0\rangle^{\otimes m}$ if $x_j=0$
for every $j$ and an $m$-qubit state orthogonal to $|0\rangle^{\otimes
m}$ if $x_j=1$ for some $j$, where $m=\lceil \log(b+1)\rceil$. Besides
the $b$ input qubits, the circuit has $m$ ancillary qubits as output
qubits. The first part of the circuit is a layer consisting of $H$ gates
on the ancillary qubits. The middle part is a quantum circuit consisting
of $b$ controlled-$R(2\pi/2^k)$ gates over all $1\leq k \leq m$, where
each gate uses an input qubit as the control qubit and an ancillary
qubit as the target qubit. Such a gate is not an elementary gate, but it
can be decomposed into a sequence of elementary gates. The last part is
the same as the first one. We call the circuit the non-commuting OR
reduction circuit. It is depicted in Fig.~\ref{figure1}(a), where
$b=3$.

An important observation is that the non-commuting OR reduction circuit
can be transformed into a 2-local commuting quantum circuit. This is
shown by considering a quantum circuit consisting of gates $g_j$ on two
qubits, where each $g_j$ is a controlled-$R(2\pi/2^k)$ gate, which is in
the non-commuting OR reduction circuit, sandwiched between Hadamard
gates on the target qubit. Since $HH=I$ and controlled-$R(2\pi/2^k)$
gates are pairwise commuting gates on two qubits, the operation
performed by the circuit is the same as that performed by the
non-commuting OR reduction circuit and the gates $g_j$ are pairwise
commuting gates on two qubits. We call the circuit the commuting OR
reduction circuit. It is depicted in Fig.~\ref{figure1}(b), where
$b=3$. Combining this commuting OR reduction circuit with $A_n$ in the
above proof implies the following lemma:

\begin{figure}[t]
\centering
\includegraphics[scale=.6]{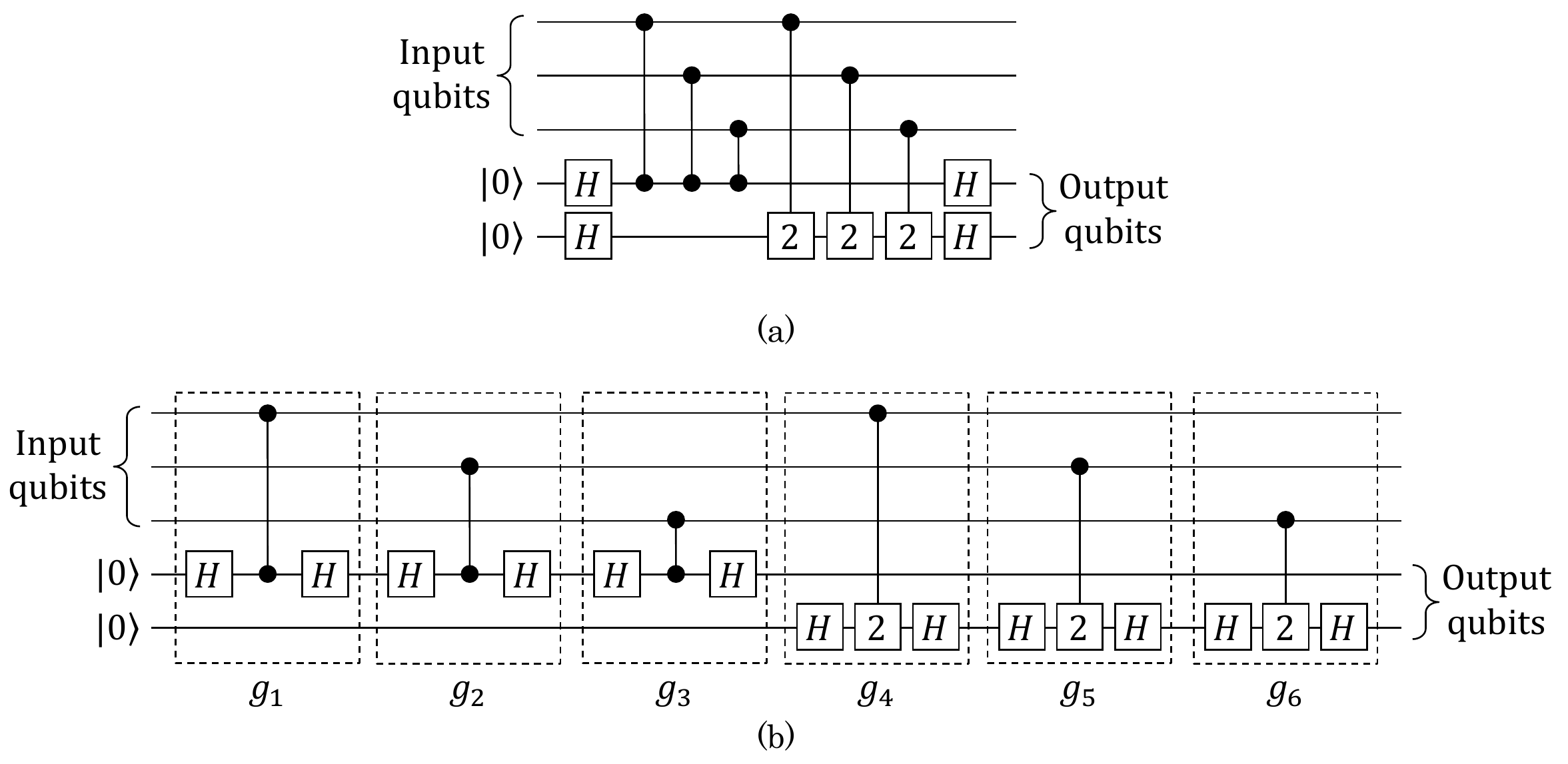}
\caption{(a) The non-commuting OR reduction circuit, where $b=3$, the
 gate represented by two black circles connected by a vertical line is a
 $\Lambda Z$ gate, i.e., a controlled-$R(2\pi/2^1)$ gate, and the gate
 represented by ``2'' is an $R(2\pi/2^2)$ gate. (b) The commuting OR
 reduction circuit, where $b=3$.}
\label{figure1}
\end{figure}

\begin{lemma}\label{depth-3+OR}
There exists a commuting quantum circuit with $O(\log n)$ output qubits
 such that it is not weakly simulatable unless $\ph$ collapses to the
 third level.
\end{lemma}
\begin{proof}
 As in the proof of Lemma~\ref{depth-3}, we can take $L \in
 \postbqp\setminus \postbpp$ and obtain a depth-3 polynomial-size
 quantum circuit $A_n$ with $n$ input qubits, $a+b$ ancillary qubits,
 and $b+2$ output qubits such that, for any $x \in \{0,1\}^n$,
\begin{itemize}
 \item if $x \in L$, ${\rm Pr}[A_n(x)=0^{b+1}1] \geq 2\cdot{\rm
       Pr[qpost}_n(x)=0^{b+1}]/3$,
 \item if $x \notin L$, ${\rm Pr}[A_n(x)=0^{b+1}1] \leq {\rm
       Pr[qpost}_n(x)=0^{b+1}]/3$.
\end{itemize}
We construct a quantum circuit $E_n$ with $n$ input qubits, $a+b+m+1$
 ancillary qubits, and $m+1$ output qubits as follows, where
 $m=\lceil\log (b+2)\rceil$. As an example, $E_n$ is depicted in
 Fig.~\ref{figure2}(a), where $n=5$, $a=0$, and $b=2$ (and thus $m=2$).
\begin{enumerate}
 \item Apply $A_n$ on $n$ input qubits and $a+b$ ancillary qubits, where
       the input qubits of $E_n$ are those of $A_n$.

 \item Apply a $\Lambda X$ gate on the last output qubit of $A_n$ and on
       an ancillary qubit (other than the ancillary qubits in Step~1),
       where the output qubit is the control qubit.

 \item Apply the commuting OR reduction circuit on the other output
       qubits of $A_n$, i.e., the $b+1$ postselection qubits of $A_n$,
       and $m$ ancillary qubits (other than the ancillary qubits in
       Steps~1 and 2), where the postselection qubits are the input
       qubits of the OR reduction circuit.

 \item Apply $A^\dag_n$ as in Step~1.
\end{enumerate}
The $m+1$ ancillary qubits used in Steps~2 and 3 are the output qubits
 of $E_n$. Step~4 does not affect the output probability distribution of
 $E_n$, but it allows us to construct the commuting quantum circuit
 described below. By the construction of $E_n$, for any $x \in
 \{0,1\}^n$,
$${\rm Pr}[A_n(x)=0^{b+1}1] = {\rm Pr}[E_n(x)=0^m1],\ {\rm
 Pr}[A_n(x)=0^{b+1}0] = {\rm Pr}[E_n(x)=0^m0].$$
This implies that $E_n$ is not weakly simulatable. The proof is the same
 as that of Lemma~\ref{depth-3} except that the number of output qubits
 we need to consider is only $m+1=O(\log n)$.

We show that there exists a commuting quantum circuit with $m+1$ output
 qubits such that its output probability distribution is the same as
 that of $E_n$. We consider a quantum circuit consisting of gates
 $A^\dag_ngA_n$ for any gate $g$ that is either a $\Lambda X$ gate in
 Step~2 of $E_n$ or $g_j$ in the commuting OR reduction circuit. The
 input qubits and output qubits of $E_n$ are naturally considered as the
 input qubits and output qubits of the new circuit, respectively. The
 circuit based on $E_n$ in Fig.~\ref{figure2}(a) is depicted in
 Fig.~\ref{figure2}(b). Since these gates $g$ in $E_n$ are pairwise
 commuting, so are the gates $A^\dag_ngA_n$. Moreover, $A^\dag_ngA_n$
 acts on a constant number of qubits (in fact, on at most $2^3+1=9$
 qubits) since the depth of $A_n$ is three, $g$ is on two qubits, and
 the number of qubits on which both $g$ and $A_n$ are applied is one. By
 the construction of the circuit, its output probability distribution is
 the same as that of $E_n$.
\end{proof}

\begin{figure}[t]
\centering
\includegraphics[scale=.6]{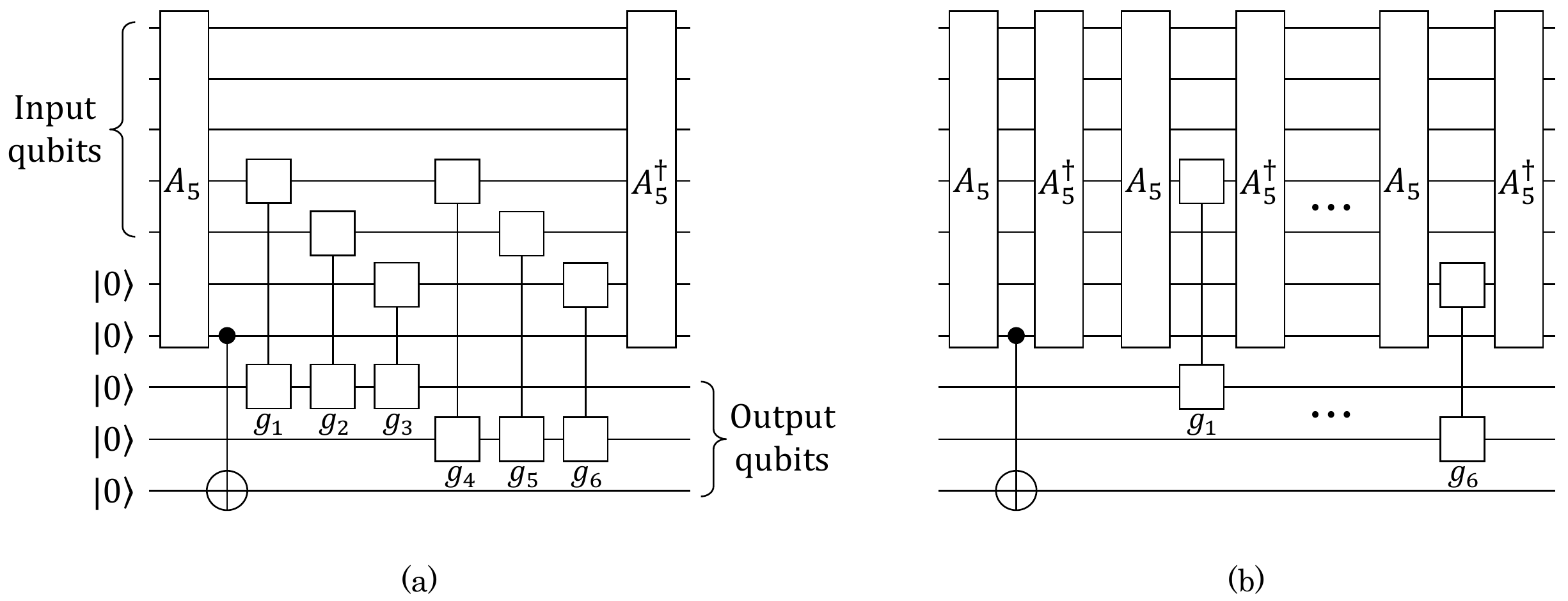}
\caption{(a) Circuit $E_n$, where $n=5$, $a=0$, and $b=2$ (and thus
 $m=2$). The gate represented by a black circle and $\oplus$ connected
 by a vertical line is a $\Lambda X$ gate. The gates $g_j$ are the ones
 in Fig.~\ref{figure1}. (b) The commuting quantum circuit based on
 $E_n$ in (a).}
\label{figure2}
\end{figure}

To complete the proof of Theorem~\ref{commuting}, it suffices to show
that the commuting quantum circuit in the proof of
Lemma~\ref{depth-3+OR} is 5-local. To show this, we give the details of
the depth-3 quantum circuit constructed by the method
in~\cite{Fenner}. The circuit is based on a one-qubit teleportation
circuit. We adopt the teleportation circuit depicted in
Fig.~\ref{figure3}(a), which is obtained from the standard one by
decomposing it into the elementary gates. If the classical outcomes of
$Z$-measurements on the two qubits other than the output qubit are 0,
the output state is the same as the input state. We call the first
measured qubit, which is the input qubit, ``the first teleportation
qubit'', and the second one ``the second teleportation qubit''.

For example, we consider the circuit depicted in Fig.~\ref{figure3}(b)
as $C_n$ in the proof of Lemma~\ref{depth-3}, where $n=2$ and $a=0$. The
depth-3 circuit $A_n$ constructed from $C_n$ by the method
in~\cite{Fenner} is depicted in Fig.~\ref{figure3}(c), where $b=6$ and
thus the total number of postselection qubits is seven. The first layer
consists of the first halves of the teleportation circuits and the third
consists of the last halves. The second layer consists of the gates in
$C_n$. The teleportation qubits are the postselection qubits. If all
classical outcomes of $Z$-measurements on the teleportation qubits are
0, all teleportation circuits teleport their input states successfully
and thus the output state is the same as that of $C_n$.

We will analyze $A^\dag_ngA_n$ in the proof of Lemma~\ref{depth-3+OR},
which implies the following lemma:
\begin{lemma}\label{5-local}
For any gate $A^\dag_ngA_n$ in the proof of Lemma~\ref{depth-3+OR},
 there exists a quantum circuit on at most five qubits that implements
 the gate.
\end{lemma}
\begin{proof}
 We first consider the case when $g=g_j$ in the commuting OR reduction
 circuit. We divide this case into the following three cases, where we
 assume that $g$ is applied on a postselection qubit $q_1$ and an output
 qubit $q_2$ of $E_n$:
\begin{itemize}
 \item Case 1: $q_1$ is the first teleportation qubit (of a
       teleportation circuit).

 \item Case 2: $q_1$ is the second teleportation qubit (of a
       teleportation circuit).

 \item Case 3: $q_1$ is the postselection qubit corresponding to the one
       of $C_n$.
\end{itemize}
We obtain the desired circuit on at most five qubits by simplifying
 $A^\dag_ngA_n$, where we represent $A_n$ as $L_3L_2L_1$, each of which
 is a layer of $A_n$. We consider Case 1 using an example of
 $A^\dag_ngA_n$ depicted in Fig.~\ref{figure4}(a), where $A_n$ is the
 circuit in Fig.~\ref{figure3}(c), $g$ is a controlled-$R(2\pi/2^k)$
 gate sandwiched between $H$ gates, and $q_1$ is the fourth qubit of
 $A_n$ from the top, which is the first teleportation qubit. By
 simplifying $L_3^\dag gL_3$, we obtain the circuit depicted in
 Fig~\ref{figure4}(b). We can further simplify the circuit and obtain
 the desired circuit on five qubits $q_1,\ldots,q_5$ depicted in
 Fig.~\ref{figure4}(c). In general, we can similarly simplify
 $A^\dag_ngA_n$ and a similar analysis works for Cases 2 and 3 and the
 case when $g=\Lambda X$. The details can be found in Appendix~A.2.
\end{proof}

\begin{figure}[t]
\centering
\includegraphics[scale=.6]{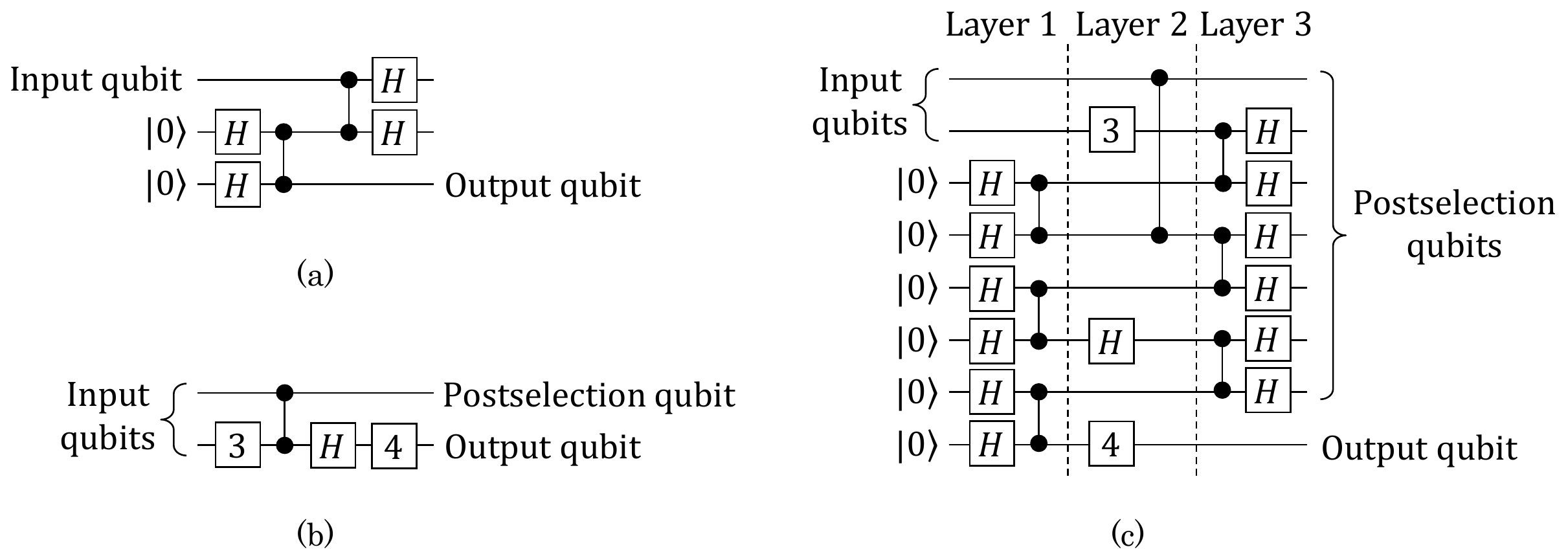}
\caption{(a) The teleportation circuit. (b) An example of circuit $C_n$,
 where $n=2$ and $a=0$. The gate represented by $k\in {\mathbb N}$ is an
 $R(2\pi/2^k)$ gate. (c) Depth-3 circuit $A_n$ constructed from $C_n$ in
 (b) by the method in~\cite{Fenner}, where $b=6$ and thus the total
 number of postselection qubits is seven.}
\label{figure3}
\end{figure}

\begin{figure}[t]
\centering
\includegraphics[scale=.6]{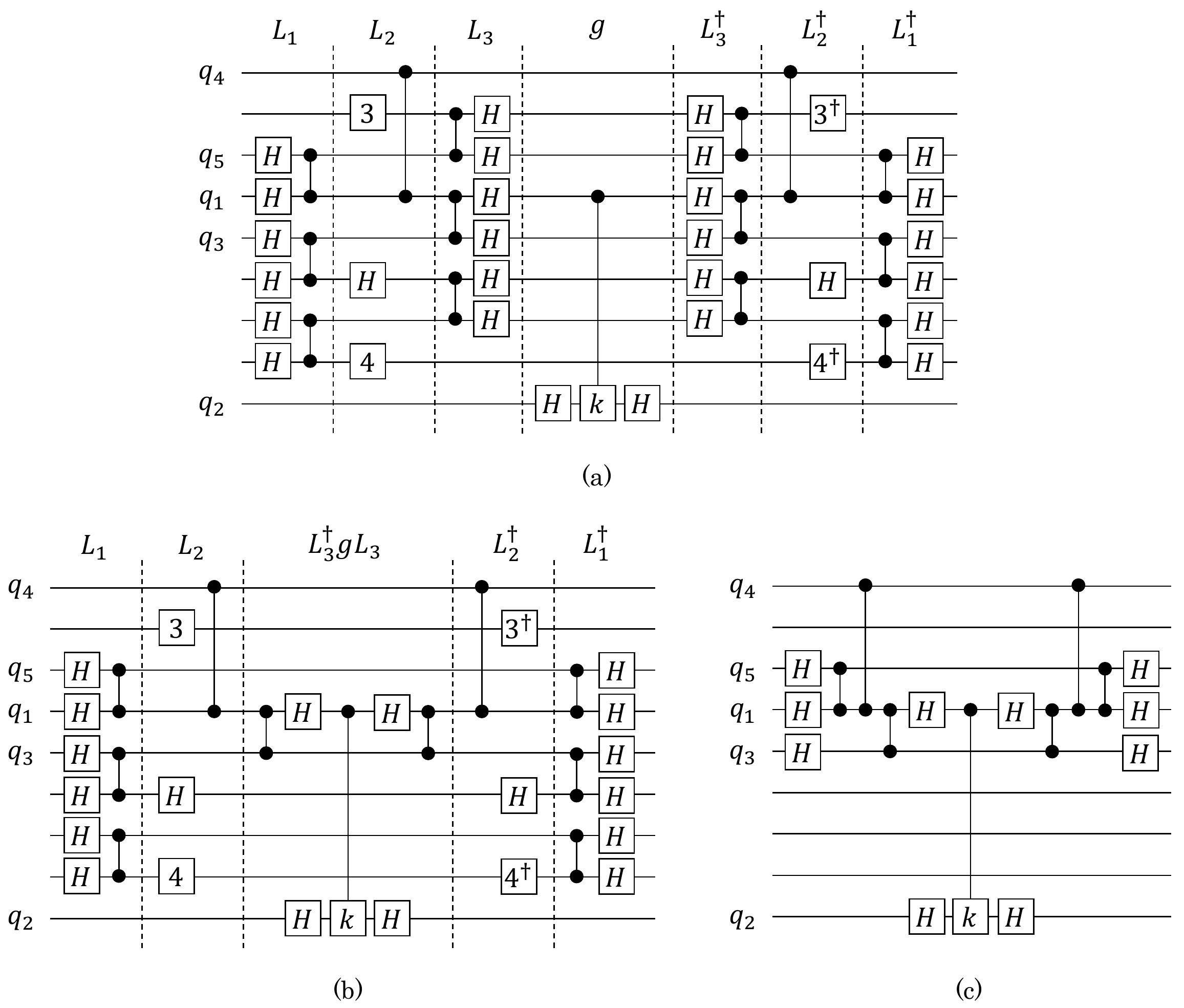}
\caption{(a) Gate $A^\dag_ngA_n$, where $A_n$ is the circuit in
 Fig.~\ref{figure3}(c), $g$ is a controlled-$R(2\pi/2^k)$ gate
 sandwiched between $H$ gates, and $q_1$ is the fourth qubit of $A_n$
 from the top. (b) The circuit obtained from $A^\dag_ngA_n$ in (a) by
 simplifying $L_3^\dag gL_3$. (c) The circuit on five qubits obtained
 from (b).}
\label{figure4}
\end{figure}

\subsection{Weak Simulatability of a Generalized Version}

The non-commuting OR reduction circuit with $b+1$ input qubits can be
represented as $H^{\otimes m}D'H^{\otimes m}$, where $m=\lceil
\log(b+2)\rceil$ and $D'$ is a quantum circuit consisting only of
controlled-$R(2\pi/2^k)$ gates. Since $\Lambda X$ is $H\Lambda ZH$, we
can represent the circuit in Theorem~\ref{commuting} as
$(A^\dag_n\otimes H^{\otimes (m+1)})D''(A_n\otimes H^{\otimes (m+1)})$,
where $D''$ consists of $D'$ and $\Lambda Z$, and $A_n$ is a depth-3
quantum circuit with $n$ input qubits, $a+b$ ancillary qubits, and $b+2$
output qubits. The output qubits of the whole circuit are the ancillary
qubits on which $H^{\otimes (m+1)}$ is applied.

We generalize the circuit in Theorem~\ref{commuting}. We assume that we
are given two quantum circuits: $F_n$ is a quantum circuit with $n$
input qubits, $s=O({\rm poly}(n))$ ancillary qubits, and $t\ (\leq n+s)$
output qubits and $D$ is a quantum circuit on $t+l$ qubits consisting of
pairwise commuting gates that are diagonal in the $Z$-basis and act on a
constant number of qubits, where $l=O(\log n)$. We consider the
following quantum circuit, which can be represented as $(F^\dag_n\otimes
H^{\otimes l})D(F_n\otimes H^{\otimes l})$, with $n$ input qubits, $s+l$
ancillary qubits, and $l$ output qubits:
\begin{enumerate}
 \item Apply $F_n$ on $n$ input qubits and $s$ ancillary qubits, where
       the input qubits of the whole circuit are those of $F_n$.

 \item Apply $H^{\otimes l}$ on $l$ ancillary qubits (other than the
       ancillary qubits in Step~1).

 \item Apply $D$ on $t+l$ qubits, which are the output qubits of $F_n$
       and the ancillary qubits in Step~2.

 \item Apply $H^{\otimes l}$ as in Step~2 and then apply $F^\dag_n$ as
       in Step~1.
\end{enumerate}
The output qubits are the ancillary qubits on which $H^{\otimes l}$ is
applied. The circuit in Theorem~\ref{commuting} corresponds to the case
when $F_n=A_n$, $D=D''$, $s=a+b$, $t=b+2$, and $l=m+1$.

When $F_n=H^{\otimes (n+s)}$ with arbitrary $s$ and $t$,
$(F^\dag_n\otimes H^{\otimes l})D(F_n\otimes H^{\otimes l})$ is weakly
simulatable~\cite{Bremner}. A simple generalization of the proof
in~\cite{Bremner} implies Theorem~\ref{sois}. In fact, we fix the state
of the qubits other than the $O(\log n)$ output qubits on the basis of
the assumption in Theorem~\ref{sois} and then follow the change of the
states of the output qubits. The details of the proof can be found in
Appendix~A.3. As described in Section~1, Theorem~\ref{sois} implies an
interesting suggestion on how to improve
Theorem~\ref{commuting}. Concretely speaking, a possible way to
construct a 3- or 4-local commuting quantum circuit that is not weakly
simulatable would be to somehow choose a depth-2 quantum circuit as
$F_n$, but such a construction is impossible.

\section{Clifford Circuits}

As an application of the construction method for the circuit in
Theorem~\ref{commuting}, we consider Clifford circuits with $n$ input
qubits, $O({\rm poly}(n))$ ancillary qubits, and $O(\log n)$ output
qubits. In this section, the ancillary qubits are allowed to be in a
general product state (not restricted to a tensor product of
$|0\rangle$). As shown in~\cite{Clark,Jozsa}, such a Clifford circuit
with only one output qubit is strongly simulatable. We first show that a
simple extension of the proof in~\cite{Clark,Jozsa} implies the strong
simulatability of a Clifford circuit with $O(\log n)$ output qubits:
\begin{lemma}\label{Clifford-strong}
Any Clifford circuit with $O({\rm poly}(n))$ ancillary qubits in a
 general product state and with $O(\log n)$ output qubits is strongly
 simulatable.
\end{lemma}
\noindent
The proof can be found in Appendix~A.4.

In contrast to Lemma~\ref{Clifford-strong}, it is known that there
exists a Clifford circuit with $n$ input qubits, $O({\rm poly}(n))$
ancillary qubits in a particular product state, and $O({\rm poly}(n))$
output qubits such that it is not weakly simulatable with respect to
multiplicative error unless $\ph$ collapses to the third
level~\cite{Jozsa}. This is shown by using the fact that any $\postbqp$
circuit can be simulated (in some sense) by a Clifford circuit. More
precisely, let $L \in \postbqp$ and $C_n$ be a polynomial-size quantum
circuit with $n$ input qubits, $a=O({\rm poly}(n))$ ancillary qubits
initialized to $|0\rangle$, one output qubit, and one postselection
qubit (and some polynomial $q$) such that, for any $x \in \{0,1\}^n$,
\begin{itemize}
 \item ${\rm Pr}[{\rm post}_n(x)=0] \geq 1/2^{q(n)}$,

 \item if $x \in L$, ${\rm Pr}[C_n(x)=1|{\rm post}_n(x)=0] \geq 2/3$,

 \item if $x \notin L$, ${\rm Pr}[C_n(x)=1|{\rm post}_n(x)=0] \leq
       1/3$.
\end{itemize}
Then, there exists a Clifford circuit $A_n$ with $n$ input qubits, $a$
ancillary qubits initialized to $|0\rangle$, $b=O({\rm poly}(n))$
ancillary qubits in a product state $|\varphi\rangle^{\otimes b}$, and
one output qubit, where $|\varphi\rangle=R(\pi/4)H|0\rangle=(|0\rangle +
e^{i\pi/4}|1\rangle)/\sqrt{2}$, such that, for any $x \in \{0,1\}^n$,
\begin{itemize}
 \item if $x \in L$, Pr$[A_n(x)=1|{\rm qpost}_n(x)=0^{b+1}] \geq 2/3$,

 \item if $x \notin L$, Pr$[A_n(x)=1|{\rm qpost}_n(x)=0^{b+1}] \leq
       1/3$,
\end{itemize}
where the event ``${\rm qpost}_n(x)=0^{b+1}$'' means that all classical
 outcomes of $Z$-measurements on the qubit corresponding to the
 postselection qubit of $C_n$ and particular $b$ qubits (other than the
 output qubit) are 0. We call these $b+1$ qubits the postselection
 qubits of $A_n$. We can show that Pr$[{\rm qpost}_n(x)=0^{b+1}] \geq
 1/2^{b+q}$. By using this property and $A_n$ obtained from $L\in
 \postbqp \setminus \postbpp$ as in the proof of Lemma~\ref{depth-3}, we
 can show the following lemma, where the classical simulatability is
 defined with respect to additive error:
\begin{lemma}
There exists a Clifford circuit with $O({\rm poly}(n))$ ancillary qubits
 in a particular product state and with $O({\rm poly}(n))$ output qubits
 such that it is not weakly simulatable unless $\ph$ collapses to the
 third level.
\end{lemma}

As in the proof of Lemma~\ref{depth-3+OR}, we construct a quantum
circuit $E'_n$ with $n$ input qubits and $a+b+m+1$ ancillary qubits by
combining $A_n$ with the non-commuting OR reduction circuit as
follows, where $m=\lceil\log (b+2)\rceil$ and the $m+1$ ancillary qubits
are the output qubits of $E'_n$. As an example, $E'_n$ is depicted in
Fig.~\ref{figure5}(a), where $n=5$, $a=0$, and $b=2$.
\begin{enumerate}
 \item Apply $A_n$ on $n$ input qubits, $a$ ancillary qubits initialized
       to $|0\rangle$, and $b$ ancillary qubits initialized to
       $|\varphi\rangle$, where the input qubits of $E'_n$ are those of
       $A_n$.

 \item Apply a $\Lambda X$ gate on the (original) output qubit of $A_n$
       and an ancillary qubit (other than the ancillary qubits in
       Step~1), where the output qubit is the control qubit.

 \item Apply the non-commuting OR reduction circuit on the $b+1$
       postselection qubits of $A_n$ and $m$ ancillary qubits (other
       than the ancillary qubits in Steps~1 and 2), where the
       postselection qubits are the input qubits of the OR reduction
       circuit.
\end{enumerate}
A direct application of the proof of Lemma~\ref{depth-3+OR} implies the
following lemma:
\begin{lemma}\label{Clifford+OR}
There exists a Clifford circuit combined with the OR reduction circuit as
 described above with $O({\rm poly}(n))$ ancillary qubits in a
 particular product state and with $O(\log n)$ output qubits such that
 it is not weakly simulatable unless $\ph$ collapses to the third level.
\end{lemma}

We replace the non-commuting OR reduction circuit in Step~3 with a
constant-depth OR reduction circuit with unbounded fan-out
gates~\cite{Hoyer}, where an unbounded fan-out gate can be considered as
a sequence of CNOT gates with the same control qubit. It is easy to show
that decomposing the unbounded fan-out gates into CNOT gates in the
constant-depth OR reduction circuit yields a Clifford-1 circuit, which
is a Clifford circuit augmented by a depth-1 non-Clifford layer. In
particular, this procedure transforms the middle part of the
non-commuting OR reduction circuit in Step~3, which is the only part
that includes non-Clifford gates, into a quantum circuit that has CNOT
gates and a depth-1 layer consisting of all gates in the middle
part. The circuit obtained from the middle part in Fig.~\ref{figure5}(a)
is depicted in Fig.~\ref{figure5}(b). This transformation with
Lemma~\ref{Clifford+OR} implies Theorem~\ref{Clifford-1}.

A similar argument implies that there exists a Clifford-2 circuit with
$O({\rm poly}(n))$ ancillary qubits in a particular product state and
with only one output qubit such that it is not weakly simulatable unless
$\ph$ collapses to the third level, where a Clifford-2 circuit has two
depth-1 non-Clifford layers. Let $L\in \postbqp \setminus \postbpp$. We
obtain $A_n$ as described above and combine it with a constant-depth
quantum circuit for the OR function with unbounded fan-out
gates~\cite{Takahashi}. By decomposing the unbounded fan-out gates into
CNOT gates, the OR circuit can be transformed into a Clifford-2
circuit. Unfortunately, a combination of the circuits similar to the
above construction has two output qubits. Thus, we construct two
circuits with one output qubit. One circuit consists of $A_n$ and the OR
circuit, where the input qubits of the OR circuit are the output qubit
of $A_n$ and $b+1$ postselection qubits, and the output qubit of the OR
circuit is the output qubit of the whole circuit. The other similarly
consists of $XA_n$ and the OR circuit, where $X$ is applied on the
output qubit of $A_n$. By a similar argument in~\cite{Takahashi-sim}, we
can show that, if these two Clifford-2 circuits are weakly simulatable,
then $L\in\postbpp$. Thus, at least one of the circuits is not weakly
simulatable.

\begin{figure}[t]
\centering
\includegraphics[scale=.6]{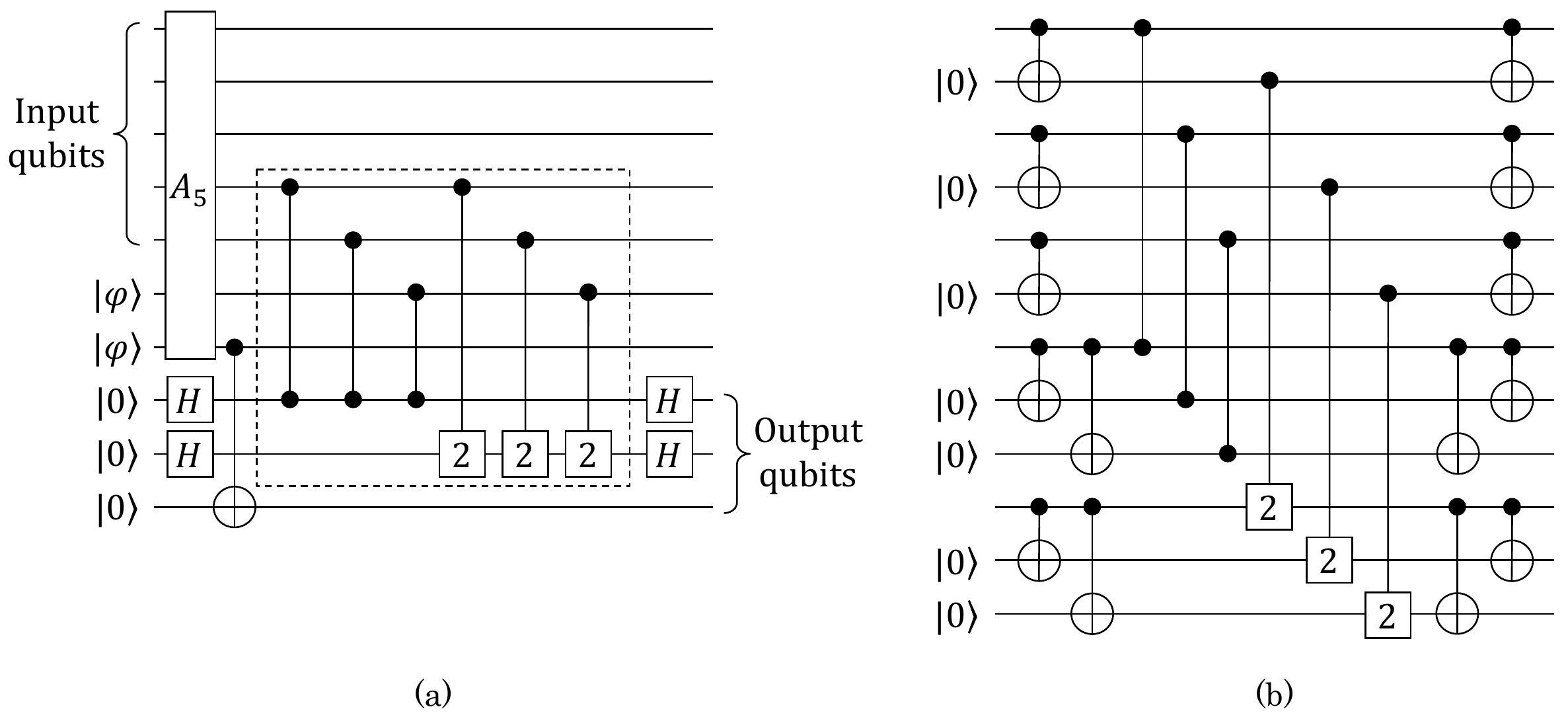}
\caption{(a) Circuit $E'_n$, where $n=5$, $a=0$, and $b=2$ (and thus
 $m=2$). The dashed box represents the middle part of the
 non-commuting OR reduction circuit. (b) The circuit that has CNOT
 gates and a depth-1 layer consisting of all gates in the middle part in
 (a). The qubits in state $|0\rangle$ are new ancillary qubits, which
 are not depicted in (a).}
\label{figure5}
\end{figure}

\section{Open Problems}

Interesting challenges would be to further investigate commuting quantum
circuits and to consider closely related computational models. Some
examples are as follows:
\begin{itemize}
 \item Does there exist a 3- or 4-local commuting quantum circuit with
       $O(\log n)$ output qubits such that it is not weakly simulatable
       (under a plausible assumption)?

 \item Do the theorems in this paper hold when exponentially small
       error $1/2^{p(n)}$ is replaced with polynomially small error
       $1/p(n)$ in the definitions of the classical simulatability?

 \item Can we apply the results on commuting quantum circuits to
       investigating the computational power of constant-depth quantum
       circuits?
\end{itemize}

\section*{Acknowledgment}

We thank Harumichi Nishimura and Tomoyuki Morimae for pointing out to us
the applicability of their proof method~\cite{Morimae}, which inspired
us to realize that a slight modification of our proof method in the
previous version of the present paper yields the stronger results
described in this version.



\bibliography{mybib}




\appendix
\section{Proofs}

\subsection{Proof of Lemma~\ref{depth-3}}

We assume that $A_n$ is weakly simulatable. Then, there exists a
 polynomial-size randomized classical circuit $R_n$ such that, for any
 $x\in\{0,1\}^n$ and $y\in\{0,1\}^{b+2}$,
$$|{\rm Pr}[R_n(x)=y]-{\rm Pr}[A_n(x)=y]| \leq \frac{1}{2^{b+q+10}}.$$
This implies that
$${\rm Pr}[A_n(x)=0^{b+1}1]-\frac{1}{2^{b+q+10}} \leq {\rm
 Pr}[R_n(x)=0^{b+1}1] \leq {\rm
 Pr}[A_n(x)=0^{b+1}1]+\frac{1}{2^{b+q+10}},$$
$${\rm Pr}[A_n(x)=0^{b+1}0]-\frac{1}{2^{b+q+10}} \leq {\rm
 Pr}[R_n(x)=0^{b+1}0] \leq {\rm
 Pr}[A_n(x)=0^{b+1}0]+\frac{1}{2^{b+q+10}}.$$
Since ${\rm Pr}[A_n(x)=0^{b+1}1]+{\rm Pr}[A_n(x)=0^{b+1}0]= {\rm
 Pr}[{\rm qpost}_n(x)=0^{b+1}]$, it holds that
\begin{align*}
 {\rm Pr}[{\rm qpost}_n(x)=0^{b+1}]-\frac{1}{2^{b+q+9}} 
&\leq {\rm Pr}[R_n(x)=0^{b+1}1] + {\rm Pr}[R_n(x)=0^{b+1}0]\\
& \leq {\rm Pr}[{\rm qpost}_n(x)=0^{b+1}] + \frac{1}{2^{b+q+9}}.
\end{align*}

We construct a polynomial-size randomized classical circuit $S_n$ that
 implements the following classical algorithm with input $x \in
 \{0,1\}^n$:
\begin{enumerate}
 \item Compute $R_n(x)$.

 \item \begin{enumerate}
	\item If $R_n(x)=0^{b+1}1$, set ${\rm post}_n(x)=0$ and
	      $S_n(x)=1$.
	\item If $R_n(x)=0^{b+1}0$, set ${\rm post}_n(x)=0$ and
	      $S_n(x)=0$.
	\item Otherwise, set ${\rm post}_n(x)=1$ and $S_n(x)=1$.
	\end{enumerate}
\end{enumerate}
By the definition of $S_n$,
\begin{align*}
{\rm Pr}[{\rm post}_n(x)=0] &= {\rm Pr}[R_n(x)=0^{b+1}1] + {\rm
 Pr}[R_n(x)=0^{b+1}0]\\
&\geq {\rm Pr}[{\rm qpost}_n(x)=0^{b+1}]-\frac{1}{2^{b+q+9}}\\
&\geq \frac{1}{2^{b+q}}-\frac{1}{2^{b+q+9}}>0.
\end{align*}
Moreover, for any $x \in \{0,1\}^n$,
$${\rm Pr}[S_n(x)=1|{\rm post}_n(x)=0] = \frac{{\rm Pr}[R_n(x)=0^{b+1}1]}
 {{\rm Pr}[R_n(x)=0^{b+1}1] + {\rm Pr}[R_n(x)=0^{b+1}0]}.$$
If $x \in L$,
\begin{align*}
{\rm Pr}[S_n(x)=1|{\rm post}_n(x)=0] &\geq 
\frac{{\rm Pr}[A_n(x)=0^{b+1}1]-\frac{1}{2^{b+q+10}}}{ {\rm Pr}[{\rm
qpost}_n(x)=0^{b+1}]+\frac{1}{2^{b+q+9}}}\\
&\geq
\frac{\frac{2}{3}\cdot{\rm Pr}[{\rm
qpost}_n(x)=0^{b+1}]-\frac{1}{2^{b+q+10}}}{ {\rm Pr}[{\rm
qpost}_n(x)=0^{b+1}]+\frac{1}{2^{b+q+9}}}\\
&=\frac{2}{3} - \frac{7\varepsilon}{3(1+2\varepsilon)} 
> \frac{2}{3} - \frac{7}{3}\varepsilon > \frac{3}{5},
\end{align*}
where $\varepsilon = 1/(2^{b+q+10}\cdot {\rm Pr}[{\rm
qpost}_n(x)=0^{b+1}])$ and it holds that
$$0 < \varepsilon \leq \frac{1}{2^{b+q+10} \cdot \frac{1}{2^{b+q}}} =
\frac{1}{2^{10}}.$$
If $x \notin L$,
\begin{align*}
{\rm Pr}[S_n(x)=1|{\rm post}_n(x)=0] &\leq
\frac{{\rm Pr}[A_n(x)=0^{b+1}1]+\frac{1}{2^{b+q+10}}}{ {\rm Pr}[{\rm
qpost}_n(x)=0^{b+1}]-\frac{1}{2^{b+q+9}}}\\
&\leq
\frac{\frac{1}{3}\cdot{\rm Pr}[{\rm
qpost}_n(x)=0^{b+1}]+\frac{1}{2^{b+q+10}}}{{\rm Pr}[{\rm
qpost}_n(x)=0^{b+1}]-\frac{1}{2^{b+q+9}}}\\
&=\frac{1}{3} + \frac{5\varepsilon}{3(1-2\varepsilon)} < \frac{2}{5}.
\end{align*}
The constants 2/3 and 1/3 in the definition of $\postbpp$ can be
 replaced with $1/2 + \delta$ and $1/2 -\delta$, respectively, for any
 constant $0 < \delta < 1/2$~\cite{Bremner}. Thus, $L \in \postbpp$.

\subsection{Proof of Lemma~\ref{5-local}}

\begin{itemize}
 \item Case 1: $q_1$ is the first teleportation qubit (of a
       teleportation circuit).

 We note that $g$ is on the set of qubits $\{q_1,q_2\}$ and that there
       is no gate on $q_2$ in each layer. All $\Lambda Z$ gates other
       than the one on $q_1$ and qubit $q_3$ in layer~3 are cancelled
       out in $L_3^\dag gL_3$. Only the $\Lambda Z$ gate, which is not
       cancelled out, increases the number of qubits involved with
       $\{q_1,q_2\}$ by one. Thus, $L_3^\dag gL_3$ is on
       $\{q_1,q_2,q_3\}$. By the construction of the teleportation
       circuit, there is no gate on $q_3$ in layer~2. Only one $\Lambda
       Z$ gate on $q_1$ and qubit $q_4$ in layer~2 increases the number
       of qubits involved with $\{q_1,q_2,q_3\}$ by one. Thus,
       $L^\dag_2L_3^\dag gL_3L_2$ is on at most four qubits. If a
       $\Lambda Z$ gate is on $q_3$ or $q_4$ and on another qubit, it is
       cancelled out in $L^\dag_1L^\dag_2L_3^\dag gL_3L_2L_1$. Only one
       $\Lambda Z$ gate on $q_1$ and qubit $q_5$ in layer~1 increases
       the number of qubits involved with $\{q_1,q_2,q_3,q_4\}$ by
       one. Thus, $L^\dag_1L^\dag_2L_3^\dag gL_3L_2L_1$ is on at most
       five qubits.

\item Case~2: $q_1$ is the second teleportation qubit (of a
teleportation circuit).

As an example, $A^\dag_ngA_n$ is depicted in Fig.~\ref{figure6}(a),
      where $A_n$ is the circuit in Fig.~\ref{figure3}(c), $g$ is a
      controlled-$R(2\pi/2^k)$ gate sandwiched between $H$ gates, and
      $q_1$ is the second qubit of $A_n$ from the bottom, which is the
      second teleportation qubit. As in Case~1, there is no gate on
      $q_2$ in each layer and $L_3^\dag gL_3$ is on
      $\{q_1,q_2,q_3\}$. The circuit obtained from $A^\dag_ngA_n$ in
      Fig.~\ref{figure6}(a) by simplifying $L_3^\dag gL_3$ is depicted
      in Fig.~\ref{figure6}(b). By the construction of the teleportation
      circuit, there is no gate on $q_1$ in layer~2. If a $\Lambda Z$
      gate is on $q_3$ and a qubit in layer~2, it is cancelled out in
      $L^\dag_2L_3^\dag gL_3L_2$. Thus, gates in layer~2 do not increase
      the number of qubits involved with $\{q_1,q_2,q_3\}$. In layer~1,
      a $\Lambda Z$ gate on $q_1$ and qubit $q_4$ increases the number
      of qubits involved with $\{q_1,q_2,q_3\}$ by one, and so does a
      $\Lambda Z$ gate on $q_3$ and qubit $q_5$. In particular, the
      latter happens only when an $H$ gate is on $q_3$ in layer~2. This
      is because, when any other gate, i.e., a $\Lambda Z$ or $R(\pm
      2\pi/2^k)$ gate, is on $q_3$ in layer~2, the gate is cancelled out
      in $L^\dag_2L_3^\dag gL_3L_2$ and thus a $\Lambda Z$ gate on $q_3$
      and qubit $q_5$ is also cancelled out in $L^\dag_1L^\dag_2L_3^\dag
      gL_3L_2L_1$. Thus, $L^\dag_1L^\dag_2L_3^\dag gL_3L_2L_1$ is on at
      most five qubits. The circuit obtained from $A^\dag_ngA_n$ in
      Fig.~\ref{figure6}(b) is depicted in Fig.~\ref{figure6}(c).

 \item Case 3: $q_1$ is the postselection qubit corresponding to the one
       of $C_n$.

       Similar to the above cases, there is no gate on $q_2$ in each
       layer. By the construction of $A_n$, there is no gate on $q_1$ in
       layer~3. Thus, it suffices to consider only $L_2L_1$. Since $g$
       is on two qubits and the number of qubits on which both $g$ and
       $L_2L_1$ are applied is one, $L^\dag_1L^\dag_2 gL_2L_1$ is on at
       most $2^2+1=5$ qubits.
\end{itemize}
The analysis for Case 3 also works for the case when $g=\Lambda X$ in
Step~2 of $E_n$.

\begin{figure}[t]
\centering
\includegraphics[scale=.6]{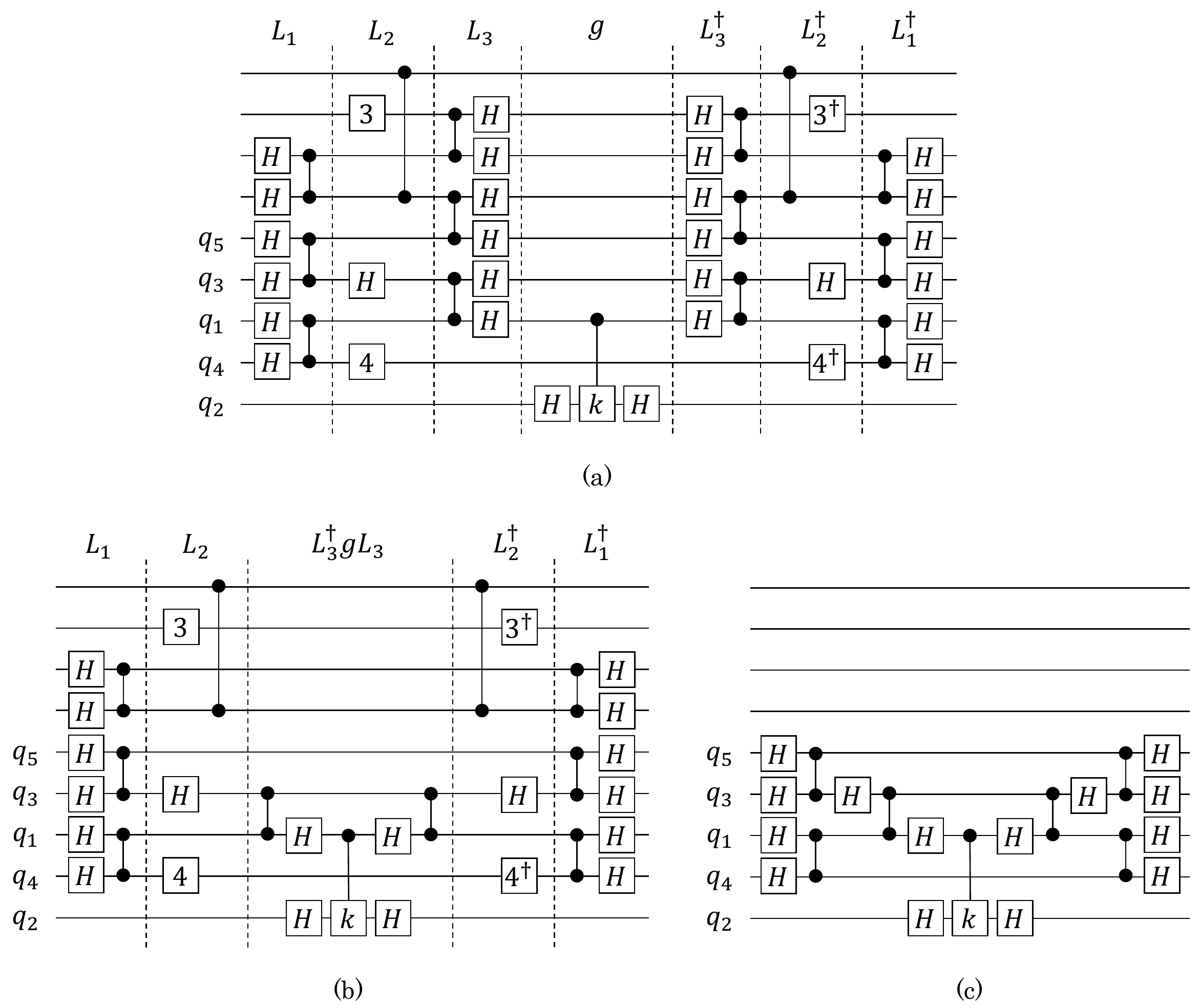}
\caption{(a) Gate $A^\dag_ngA_n$, where $A_n$ is the circuit in
 Fig.~\ref{figure3}(c), $g$ is a controlled-$R(2\pi/2^k)$ gate
 sandwiched between $H$ gates, and $q_1$ is the second qubit of $A_n$
 from the bottom. (b) The circuit obtained from $A^\dag_ngA_n$ in (a) by
 simplifying $L_3^\dag gL_3$. (c) The circuit on five qubits obtained
 from (b).}
\label{figure6}
\end{figure}

\subsection{Proof of Theorem~\ref{sois}}

Let $|x\rangle$ be an $n$-qubit input state, where $x \in
 \{0,1\}^n$. Moreover, let 
$$F_n|x\rangle|0\rangle^{\otimes s}=\sum_{z\in
 \{0,1\}^{t}}\alpha_{x,z}|z\rangle|\psi_{x,z}\rangle,$$
where $\alpha_{x,z} \in {\mathbb C}$ and $|\psi_{x,z}\rangle$ is an
 $(n+s-t)$-qubit state. Then,
$$(F^\dag_n\otimes H^{\otimes l})D(F_n\otimes H^{\otimes
 l})|x\rangle|0\rangle^{\otimes (s+l)} =
 \frac{1}{\sqrt{2^l}}(F^\dag_n\otimes H^{\otimes l}) \sum_{z\in
 \{0,1\}^{t},w\in \{0,1\}^l}\alpha_{x,z}D|z\rangle |\psi_{x,z}\rangle
 |w\rangle.$$
Since $D$ consists only of gates that are diagonal in the $Z$-basis,
 $D|z\rangle|w\rangle=e^{if(z,w)}|z\rangle |w\rangle$ for some value
 $f(z,w)$ computed from the diagonal elements of $D$. Thus, the above
 state is
$$\frac{1}{\sqrt{2^l}}(F^\dag_n\otimes H^{\otimes l})
\sum_{z\in \{0,1\}^{t},w\in
 \{0,1\}^l}\alpha_{x,z}e^{if(z,w)}|z\rangle|\psi_{x,z}\rangle
 |w\rangle.$$
Thus, for any $y\in\{0,1\}^l$, the probability that $(F^\dag_n\otimes
 H^{\otimes l})D(F_n\otimes H^{\otimes l})$ outputs $y$, which is
 represented as
$${\rm Pr}[(F^\dag_n\otimes H^{\otimes l})D(F_n\otimes H^{\otimes
 l})(x)=y],$$
is computed as
\begin{align*}
&\frac{1}{2^l} \sum_{z,z'\in \{0,1\}^{t},w,w'\in \{0,1\}^l}
 \alpha^\dag_{x,z'}\alpha_{x,z}e^{-if(z',w')+if(z,w)} \langle
 z'|z\rangle \langle\psi_{x,z'}|\psi_{x,z}\rangle \langle w'|H^{\otimes
 l}|y\rangle\langle y| H^{\otimes l}|w\rangle\\
&= \sum_{z \in \{0,1\}^{t}}|\alpha_{x,z}|^2\cdot \frac{1}{2^l}
 \sum_{w,w'\in \{0,1\}^l} e^{-if(z,w')+if(z,w)}
 \langle w'| H^{\otimes l}
|y\rangle\langle y| H^{\otimes l}|w\rangle.
\end{align*}
Let $p(n)$ be an arbitrary polynomial. By the assumption, there exists a
 polynomial-size randomized classical circuit $R_n$ such that, for any
 $x\in\{0,1\}^n$ and $z\in\{0,1\}^{t}$,
$$|{\rm Pr}[R_n(x)=z]-{\rm Pr}[F_n(x)=z]|=|{\rm
 Pr}[R_n(x)=z]-|\alpha_{x,z}|^2| \leq \frac{1}{2^{p(n)+t}}.$$
We consider a polynomial-size randomized classical circuit $T_n$ that
 implements the following classical algorithm for generating the
 probability distribution
$$\{(y,{\rm Pr}[(F^\dag_n\otimes H^{\otimes l})D(F_n\otimes H^{\otimes
 l})(x)=y])\}_{y\in \{0,1\}^l},$$ 
where the input is $x \in \{0,1\}^n$:
\begin{enumerate}
 \item Compute $z_0=R_n(x)\in\{0,1\}^{t}$.

 \item Compute the probability that $Z$-measurements on the state
$$\frac{1}{\sqrt{2^l}} \sum_{w\in
       \{0,1\}^l}e^{if(z_0,w)}H^{\otimes l}|w\rangle$$
       output $y$ for any $y \in \{0,1\}^l$.
 \item Output $y\in\{0,1\}^l$ according to the probability
       distribution computed in Step~2.
\end{enumerate}

The probability in Step~2 is represented as
$$\frac{1}{2^l} \sum_{w,w'\in \{0,1\}^l}e^{-if(z_0,w')+if(z_0,w)}
 \langle w'| H^{\otimes l} |y\rangle\langle y| H^{\otimes l}|w\rangle.$$
We can compute $f(z_0,w)$ using a polynomial-size classical
circuit since $D$ has polynomially many gates $g$ and it is easy to
classically compute $\gamma_g\in {\mathbb C}$ such that
$g|z_0\rangle|w\rangle=\gamma_g|z_0\rangle|w\rangle$ by using the
classical description of $D$, which includes information about the
complex numbers defining $g$ and the qubit numbers on which $g$ is
applied. Moreover, since the state in Step~2 is only on $l=O(\log n)$
qubits, we can compute the probability in Step~2 up to an exponentially
small additive error using a polynomial-size classical circuit. In the
following, for simplicity, we assume that we can compute the probability
exactly. Then, for any $y\in\{0,1\}^l$,
$${\rm Pr}[T_n(x)=y] = \sum_{z_0\in\{0,1\}^{t}}{\rm Pr}[R_n(x)=z_0]\cdot
\frac{1}{2^l} \sum_{w,w'\in \{0,1\}^l}e^{-if(z_0,w')+if(z_0,w)} \langle
w'| H^{\otimes l} |y\rangle\langle y| H^{\otimes l}|w\rangle.$$
This implies that, for any $x\in\{0,1\}^n$ and $y\in\{0,1\}^l$,
\begin{align*}
 |{\rm Pr}[T_n(x)=y]-{\rm Pr}[(F^\dag_n\otimes H^{\otimes
 l})D(F_n\otimes H^{\otimes l})(x)=y]|
& \leq \sum_{z_0\in\{0,1\}^{t}}|{\rm
 Pr}[R_n(x)=z_0]-|\alpha_{x,z_0}|^2|\\
& \leq \frac{2^{t}}{2^{p(n)+t}}=\frac{1}{2^{p(n)}}.
\end{align*}
A similar argument works when we compute the probability in Step~2 up to 
an exponentially small additive error. Thus, $(F^\dag_n\otimes
H^{\otimes l})D(F_n\otimes H^{\otimes l})$ is weakly simulatable.

\subsection{Proof of Lemma~\ref{Clifford-strong}}

Let $C_n=G_N\cdots G_1$ be a Clifford circuit with $n$ input qubits,
 $a=O({\rm poly}(n))$ ancillary qubits, and $l=O(\log n)$ output qubits,
 where $N=O({\rm poly}(n))$ and $G_j$ is $H$, $P$, or $\Lambda Z$. For
 any $x=x_1\cdots x_n\in \{0,1\}^n$, let
 $|\psi_x\rangle=|x_1\rangle\cdots |x_n\rangle |\psi_1\rangle\cdots
 |\psi_a\rangle$ be an input state, where $|\psi_j\rangle$ is a
 one-qubit state. For any $y\in\{0,1\}^l$,
$${\rm Pr}[C_n(x)=y]=\langle\psi_x|C^\dag_n|y\rangle\langle
 y|C_n|\psi_x\rangle = \langle\psi_x|C^\dag_nX_y|0\rangle^{\otimes
 l}\langle 0|^{\otimes l}X_yC_n|\psi_x\rangle,$$
where $X_y$ is the tensor product of $X$ and $I$ such that
 $X_y|0\rangle^{\otimes l}=|y\rangle$. As described in~\cite{Ni}, it can
 be shown by induction on $l$ that
$$|0\rangle^{\otimes l}\langle 0|^{\otimes l} =
 \frac{1}{2^l}\sum_{S\subseteq \{1,\ldots,l\}}Z(S),$$
where $Z(S)$ is the tensor product of $Z$ and $I$ such that $Z$ is only
 on qubit $j\in S$. Thus,
$${\rm Pr}[C_n(x)=y]=\frac{1}{2^l}\sum_{S\subseteq \{1,\ldots,l\}}
\langle\psi_x|G^\dag_1\cdots G^\dag_NX_yZ(S)X_yG_N\cdots
 G_1|\psi_x\rangle.$$
We can represent $G^\dag_NX_yZ(S)X_yG_N$ as a tensor product of Pauli
 gates with some coefficient $\pm1$ since $G_N$ is a Clifford gate and
 $X_yZ(S)X_y$ is a tensor product of Pauli gates (in fact, $Z$ and $I$
 gates with some coefficient $\pm1$). We repeat this transformation $N$
 times and obtain
\begin{align*}
{\rm Pr}[C_n(x)=y] & = \frac{1}{2^l}\sum_{S\subseteq
 \{1,\ldots,l\}}\gamma^S \langle\psi_x|P_1^S\otimes \cdots \otimes
 P_{n+a}^S|\psi_x\rangle \\
& = \frac{1}{2^l}\sum_{S\subseteq \{1,\ldots,l\}}\gamma^S
 \langle x_1|P_1^S|x_1\rangle \cdots
 \langle x_n|P_n^S|x_n\rangle
 \langle\psi_1|P_{n+1}^S|\psi_1\rangle \cdots
 \langle\psi_a|P_{n+a}^S|\psi_a\rangle
\end{align*}
for some coefficient $\gamma^S$ and Pauli gates $P_j^S$. It is easy to
 construct a polynomial-time classical algorithm for obtaining
 $\gamma^S$ and $P_j^S$ for any $S\subseteq \{1,\ldots,l\}$. Moreover,
 since $l=O(\log n)$, it suffices to consider only polynomially many
 $S$. Thus, the above representation immediately implies a
 polynomial-time classical algorithm for computing Pr$[C_n(x)=y]$. The
 marginal output probability can also be computed similarly and thus
 $C_n$ is strongly simulatable.
\end{document}